\documentclass[11pt,a4paper]{amsart}
\usepackage{fullpage}

\usepackage{hyperref,amsthm,amsfonts,amsmath, amssymb, graphicx,float,graphicx, enumerate, algorithmic, algorithm, color}
\usepackage[numbers,sort&compress]{natbib}

\theoremstyle{plain}
\newtheorem{theorem}{Theorem}[section]
\newtheorem{lemma}[theorem]{Lemma}

\newtheorem{claim}[theorem]{Claim}

\theoremstyle{definition}

\renewcommand{\l}{\ell}

\newcommand{\N}{\mathbb{N}}
\newcommand{\DEF}{\sl}       

\newcommand{\eps}{\varepsilon}

\renewcommand{\leq}{\leqslant}
\renewcommand{\le}{\leqslant}
\renewcommand{\geq}{\geqslant}
\renewcommand{\ge}{\geqslant}

\newcommand{\rep}{rank reduction problem}

\newcommand{\pvc}{partial vertex cover}
\newcommand{\mvc}{maximum vertex cover}

\begin{document}

\title{Reducing the rank of a matroid}

\author{Gwena\"el Joret}
\address{\newline D\'epartement d'Informatique
\newline Universit\'e Libre de Bruxelles
\newline Brussels, Belgium}
\email{gjoret@ulb.ac.be}

\author{Adrian Vetta}
\address{\newline Department of Mathematics and Statistics, and School of Computer Science 
\newline McGill University
\newline Montreal, Canada}
\email{vetta@math.mcgill.ca}

\sloppy
\maketitle

\begin{abstract}
We consider the {\em \rep} for matroids: 
Given a matroid $M$ and an integer $k$, find a minimum size subset of elements of $M$ whose removal reduces the rank of $M$ by at least $k$. 
When $M$ is a graphical matroid this problem is the minimum $k$-cut problem,
which admits a $2$-approximation algorithm.
In this paper we show that the {\rep}
for transversal matroids is essentially at least as hard to approximate 
as the densest $k$-subgraph problem. 
We also prove that, while the problem is easily solvable in polynomial time for partition matroids,
it is NP-hard when considering the intersection of two partition matroids. Our proof shows,
in particular, that the maximum vertex cover problem
is NP-hard on bipartite graphs, which answers an open problem of B.\ Simeone.
\end{abstract}

\section{Introduction}

Consider the well-known minimum $k$-cut problem: Given a graph $G$ and an integer $k$, find a minimum size
subset of edges whose removal increases the number of connected components by at least $k$. 
This problem is NP-hard, assuming $k$ is part of the input, and  
several $2$-approximation algorithms have been developed for it over the 
years~\cite{NR01, RS02, SV95, ZNI01}. 
Notice that the minimum $k$-cut problem has a simple formulation in terms of matroids: 
Given a graph $G$ and an integer $k$, find a minimum size subset of elements of the graphical matroid 
of $G$ whose removal reduces its rank by at least $k$.

This observation motivates the study of the {\em rank reduction} problem in other classes of matroids.
For example, is the {\rep} computationally hard and, if so, does it admit approximation algorithms with 
good approximation guarantees (as is the case for graphical matroids)? 
Moreover, as we will see, many fundamental problems can be formulated in this rank reduction framework.

In this paper, our focus is on the case of transversal matroids. 
First, we show that the {\rep} in transversal matroids is roughly at least as hard to approximate as the densest $k$-subgraph problem: 
Given a graph $G$ and an integer $k$, find a subset of $k$ vertices inducing a subgraph with a maximum number of edges. 
(Note that in our reduction the parameter $k$  is not necessarily the same one as in the {\rep}.) 
The densest $k$-subgraph problem can be approximated to within a factor of 
$O(n^{\frac14 +\epsilon})$ due to a recent breakthrough result of \cite{BCC10}.
Moreover, it is widely believed  \cite{BCG12,AC09,FDK01,FS97} that 
its hardness  is also close to this upper bound --
indeed, Bhaskara et al.~\cite{BCG12} present $n^{\Omega(1)}$ lower bounds for lift and project 
methods based upon the Sherali-Adams and the Lassere hierarchies.
We will show in particular that an $O(m^{\varepsilon})$-approximation algorithm
for the {\rep} on transversal matroids (where $m$ denotes the number of elements)
implies an $O(n^{4\varepsilon})$-approximation algorithm for the densest $k$-subgraph (where $n$ is the number of vertices).

Secondly, we prove that while the rank reduction 
problem is easily solvable in polynomial time for partition matroids (a special class of transversal matroids),
it is NP-hard when considering the intersection of two partition matroids. Our proof shows
in particular that the maximum vertex cover problem---also known as the partial vertex cover problem---is NP-hard on bipartite graphs. Here, one is given a graph and a positive integer $k$, and the goal is to 
find a set of $k$ vertices hitting as many edges as possible. 
The problem is obviously NP-hard on arbitrary graphs since it contains the vertex cover problem 
as a special case. Whether it remained NP-hard on bipartite graphs 
was an open problem of B.\ Simeone (see~\cite{EgresOpen}). 
We note that we learned after finishing this paper that Apollonio and Simeone~\cite{APsub} independently obtained a proof of this result.

\section{Preliminaries}
\label{sec:preliminaries}

In this section we give the necessary definitions and notations. 
All graphs and matroids in this paper are finite, 
and ``graph" will always mean an undirected simple 
graph. We use the shorthands $|G|$ and $||G||$ for the number of vertices and edges of 
a graph $G$, respectively. We denote by $\mu(G)$ the maximum size of a matching in $G$, which 
we call the {\DEF matching number} of $G$. 

A {\DEF matroid} $M$ is a pair $(E, \mathcal{I})$ where 
$\mathcal{I}$ is a family of subsets, called the {\DEF independent sets}, of the 
{\DEF ground set} $E$ satisfying the following three axioms: 
\begin{itemize}
\item the empty set $\varnothing$ is  independent; 
\item every subset of an independent is again independent, and
\item if $X$ and $Y$ are two independent sets with $|X| > |Y|$ then there exists 
$x \in X \setminus Y$ such that $Y\cup \{x\}$ is independent. 
\end{itemize}
The inclusion-wise maximal independent sets are the {\DEF bases} of the matroid $M$; 
as follows from the third axiom the bases all have the same cardinality. 
The {\DEF rank function} of $M$ is the function $r: 2^{E} \to \N$ that assigns 
to each subset $X$ of elements of $E$ the maximum size $r(X)$ of an independent set 
contained in $X$, called the {\DEF rank} of $X$. 
In particular, $r(E)$ is the cardinality of a basis of $M$, which is 
called the {\DEF rank} of $M$.

The {\DEF rank reduction} problem for matroids is defined as follows: 
Given a positive integer $k$ and a matroid on a set $A$ of elements 
with rank function $r$, the goal is to find a minimum size 
subset $X \subseteq A$ such that $r(A \setminus X) \leq r(A) - k$. 

For example consider the case of graphical matroids: Given a graph $G=(V, E)$, the
{\DEF graphical matroid} of $G$ is obtained by taking $E$ as ground set, and letting 
a subset $F$ of edges be independent if and only if the corresponding subgraph is acyclic. 
Here the rank reduction problem is the minimum $k$-cut problem.

As stated, here we study transversal matroids.
A bipartite graph $G$ with bipartition $(A, B)$ induces a matroid 
$M$ as follows: The matroid $M$ has $A$ as ground set, and 
$X\subseteq A$ is independent in $M$ if and only if there exists a matching of $G$ 
covering $X$. The fact that this is indeed a matroid is well-known; 
see for instance~\cite{Sch03}. Any matroid $M$ that can be obtained this way is 
called a {\DEF transversal matroid}, and the bipartite graph $G$ is said to be a 
{\DEF model} for $M$. Observe that, letting $r$ 
denote the rank function of $M$, the rank $r(X)$ of $X\subseteq A$ 
is equal to $\mu(G[X \cup B])$. 
Also note that being a transversal matroid is a hereditary property, 
in the sense that for each set $X\subseteq A$, 
taking the restriction $I \cap X$ of all independent sets $I$ yields a transversal matroid 
on ground set $X$. 

A special case of transversal matroids are partition matroids. Here we are
given a collection $E_{1}, \dots, E_{p}$ of disjoint sets and integers $d_{1}, \dots, d_{p}$ 
such that $0 \leq d_{i} \leq |E_{i}|$ for each $i \in \{1, \dots, k\}$. One can define 
a corresponding matroid with ground set $E:= E_{1} \cup \cdots \cup E_{k}$ by letting 
$X \subseteq E$ be independent if and only if $|E_{i} \cap X| \leq d_{i}$ for each 
$i \in \{1, \dots, k\}$. Such a matroid is called a {\DEF partition matroid}, with {\DEF model}
 $((E_{1}, d_{1}), \dots, (E_{p}, d_{p}))$. 
This corresponds to a transversal matroid on a
bipartite graph $G$ with bipartition $(E, B)$, where $B$ has $d_{i}$ vertices that are adjacent 
to all vertices in $E_{i}$, and none other, for each set $E_{i}$.  
Notice that partition matroids are also hereditary. 

Throughout, since we restrict ourselves to specific families of matroids, 
we assume that the matroid is given concisely and not given explicitly as a set system in input.
Specifically, a corresponding model of the matroid is provided: a bipartite graph for a transversal matroid, 
a graph for a graphical matroid, etc.

More generally, the rank reduction problem can be considered on the intersection 
of matroids. Given two matroids $M_{1}=(E, \mathcal{I}_{1})$ and $M_{2}=(E, \mathcal{I}_{2})$ with  
common ground set $E$, the intersection $M_{1}\cap M_{2}$ of $M_{1}$ and $M_{2}$ is the pair 
$(E, \mathcal{I})$ where $\mathcal{I}$ is the family of sets $X\subseteq E$ that are independent in 
both $M_{1}$ and $M_{2}$, which are said to be the independent sets of $M_{1}\cap M_{2}$.  
While the independence system $(E, \mathcal{I})$ is not necessarily a matroid anymore, 
it enjoys several of the nice properties of matroids (see~\cite{Sch03}).  
In particular, letting as before 
the {\DEF rank} $r(X)$ of $X\subseteq E$ be the maximum size of an independent set of 
$M_{1} \cap M_{2}$ contained in $X$, 
the rank $r(E)$ of $M_{1} \cap M_{2}$ 
can be computed in polynomial time given access to the rank functions $r_{1}$ and $r_{2}$ of 
$M_{1}$ and $M_{2}$, respectively, by a classical result of Edmonds (see~\cite{Sch03}).   
We examine the rank reduction problem for the intersection of two partition matroids in 
Section~\ref{sec:mvc}.

\section{Transversal Matroids}
\label{sec:transversal_matroids}

We start our investigation of the rank reduction problem 
with an easy observation, namely that the problem can be solved in polynomial time 
if the input matroid is a partition matroid.  

\begin{theorem}
\label{thm:partition}
The rank reduction problem can be solved in polynomial time on partition matroids. 
\end{theorem}
\begin{proof}
Let $M$ be a given partition matroid with model 
$((E_{1}, d_{1}), \dots, (E_{p}, d_{p}))$ and rank function $r$. Let $E:=E_{1} \cup \cdots \cup E_{p}$ 
denote the ground set of $M$. Observe that $r(E) = \sum_{i=1}^{p}d_{i}$. 
Let $k$ be the given parameter for the rank reduction problem on $M$.  
We may assume that $1 \leq k \leq r(E)$. Let $c_{i} := |E_{i}| - d_{i}$ for each $i\in \{1, \dots, p\}$. 
Given $X \subseteq E$, the rank $r(E \setminus X)$ of the set $E \setminus X$ is 
equal to $\sum_{i=1}^{p} \min(|E_{i} \setminus X|, d_{i})$. 

Let $X\subseteq E$ be such that $r(E \setminus X) \leq r(E) - k$. Moreover, assume $X$ is inclusion-wise 
minimal with this property. Then, for each $i\in \{1, \dots, p\}$, either 
$|E_{i} \cap X| \geq c_{i} + 1$ or $E_{i} \cap X = \varnothing$. Moreover, 
letting $J$ be the subset of indices $i\in \{1, \dots, p\}$ such that
$E_{i} \cap X \neq \varnothing$, we have that $\sum_{i\in J} d_{i} \geq k$ and 
$|X| = k + \sum_{i\in J} c_{i}$. 

Conversely, suppose $J' \subseteq \{1, \dots, p\}$ is such that 
$\sum_{i\in J'} d_{i} \geq k$. Then choosing arbitrarily 
$c_{i}$ elements of $E_{i}$, for each $i\in J'$, plus $k$ additional elements from 
$\bigcup_{i \in J'} E_{i}$ gives a set $X'$ 
with $|X'| = k + \sum_{i\in J'} c_{i}$ such that $r(E \setminus X') \leq r(E) - k$.
 
Therefore, computing an optimal solution to the rank reduction problem reduces to the problem of 
finding a subset $J \subseteq \{1, \dots, p\}$ such that $\sum_{i\in J} d_{i} \geq k$ 
and $\sum_{i\in J} c_{i}$ is minimum. Thus we obtain a knapsack problem.
Moreover, as $c_{i}$ and $d_{i}$ are at most $|E_i|$, they are of polynomial size when encoded in unary.
Thus the knapsack problem can be solved easily in polynomial time using dynamic programming. 
\end{proof}

While the rank reduction problem admits a simple polynomial-time algorithm on partition matroids, 
the problem turns out to be more difficult on the broader class of transversal matroids. 
In fact, up to some degree, the problem can be viewed as a generalization of the
{\DEF densest $k$-subgraph} problem. In the latter problem, one is 
given a graph $G$ and a positive integer $k$, and the aim is to find a subgraph $H$ of $G$
with $|H|=k$ and $||H||$ maximum. 
Towards this goal, we consider a closely related problem, the {\DEF minimum $t$-edge subgraph} problem:
Given a graph $G$ and a positive integer $t$, the goal is to find a subgraph $H$ of $G$
with $||H||=t$ and $|H|$ minimum. 

We start by drawing a connection between the rank reduction problem on transversal matroids and the
 minimum $t$-edge subgraph problem. Then we will extend the connection to the densest $k$-subgraph problem.
\begin{lemma}
\label{lem:to_t_edge}
For each constant $\eps$ with $0 < \eps < 1/2$, 
every $O(m^{\eps})$-approximation algorithm for the rank reduction problem on transversal matroids 
with $m$ elements can be turned into an $O(n^{2\eps})$-approximation algorithm for the minimum $t$-edge 
subgraph problem on graphs with $n$ vertices.  
\end{lemma}
\begin{proof}
Let $G = (V, E)$ be an instance of the minimum $t$-edge subgraph problem. 
Let $n:= |V|$. 
Let $V_{1}, V_{2}, \dots, V_{n}$ be $n$ disjoint copies of $V$.
Let $E'$ be a disjoint copy of $E$.
Let $H$ be the bipartite graph with bipartition $(A,B)$ where
$A = V_{1} \cup V_{2} \cup \cdots \cup V_{n} \cup E'$ and $B=E$, 
and where $u\in A$ is adjacent to $v \in B$
if either $u$ corresponds to a vertex of $G$ that is incident to the edge corresponding to $v$ in $G$, 
or if $u$ and $v$ correspond to the same edge of $G$.

Let $r$ denote the rank function of the transversal matroid induced by $H$ on $A$; 
thus for $X\subseteq A$, $r(X)$ is the maximum size of a matching in $H[X \cup B]$.
Obviously, $r(A) = |E|$, since every $v\in B$ can be matched to its copy in $A$. 
Let $m:=|A|$ denote the number of elements of the transversal matroid.  
Now consider the rank reduction problem on this matroid with $k=t$.
Recall that a feasible solution is a subset $X \subseteq A$ 
such that $r(A \setminus X) \leq r(A) - t = |E| - t$.

As is well known,  
we have that $r(A \setminus X) \leq r(A) - t$ for $X \subseteq A$ 
if and only if there exists $Y \subseteq B$ such that $|N_{H}(Y) \setminus X| \leq |Y| - t$, 
where $N_{H}(Y)$ denotes the set of vertices of $H$ 
that have a neighbor in $Y$.\footnote{This is a consequence of Hall's Marriage Theorem, 
as we now explain for completeness.    
Add $t-1$ new vertices to the set $A$, yielding a set $A'$, and make each of them 
adjacent to every vertex in $B$. Let $H'$ be the bipartite graph obtained from $H$ in this
manner.  
Let $X \subseteq A$. Then every matching $M$ in $H \setminus X$ with 
$|M| \leq |B| - (t-1)$ can be extended to a matching $M'$ of $H' \setminus X$ 
with $|M'| = |M| + t -1$. Conversely, 
every matching $M'$ in $H' \setminus X$ with 
$|M'| \geq t-1$  yields a matching $M$ of $H \setminus X$ 
with $|M| \geq |M'| - (t -1)$ by discarding the at most $t-1$ edges of $M'$ incident to 
the vertices in $A' \setminus A$. Hence, $H \setminus X$ has a matching of size 
$|B| - (t-1)$---or equivalently, $r(A \setminus X) \geq r(A) - (t-1)$---if 
and only if $B$ can be completely matched in $H' \setminus X$. 
By Hall's theorem, the latter happens 
if and only if $|N_{H' - X}(Y)| \geq |Y|$ for every $Y \subseteq B$, 
which is equivalent to $|N_{H}(Y) \setminus X| \geq |Y| - (t-1)$ for every $Y \subseteq B$.    
Therefore, $r(A \setminus X)  \leq r(A) - t$ if and only if there exists $Y \subseteq B$ such that 
$|N_{H}(Y) \setminus X| \leq |Y| - t$. 
} 
Such a set $Y$ is said to be a {\DEF witness} for $X$. 
The set $Y$ defines in turn a corresponding subgraph $G_{Y}$ of $G$ consisting of all 
the edges of $G$ included in $Y$, and the vertices of $G$ incident to those edges. 
By definition of $H$, the set $N_{H}(Y)$ consists of the $n$ copies of each vertex of 
$G_{Y}$, along with the copies in $E'$ of each edge of $G_{Y}$. 
Observe that any set $X'$ obtained by taking the $n$ copies in $A$ of each vertex of $G_{Y}$ and
$t$ arbitrarily chosen edges of $G_{Y}$ in $A$ is such that 
$|N_{H}(Y) \setminus X'| = |Y| - t$. Moreover, since $|N_{H}(Y) \setminus X| \leq |Y| - t$ and 
$X' \subseteq N_{H}(Y)$, it follows that $|X'| \leq |X|$, 
that is, $X'$ is a solution of size no greater than $X$ and having the same witness $Y$. 
Such a pair $(X', Y)$ is called a {\DEF canonical pair}.

Now, if a canonical pair $(X, Y)$ is such that $|Y| > t$, then 
$N_{H}(Y) \setminus X$ 
consists of exactly $|Y| -t > 0$ edges of $G_{Y}$ (or more precisely, 
their copies in $E'$). 
For each $u\in N_{H}(Y) \setminus X$ with corresponding copy $v$ in $B$, 
we have that $X \setminus \{u\}$ is again a solution to the rank reduction problem, with 
witness $Y \setminus \{v\}$, and of size smaller than $X$.  

To summarize the above discussion, 
given an arbitrary set $X' \subseteq A$ such that $r(A \setminus X') \leq r(A) - t$, 
one can in polynomial time compute a canonical pair $(X, Y)$ with $|Y|=||G_{Y}||=t$ 
and $|X| = n|G_{Y}| + t \leq |X'|$.  

Conversely, for every subgraph $G' \subseteq G$ with $||G'|| = t$, there is a natural 
corresponding canonical pair $(X, Y)$, where $Y$ contains the copies in $B$ of the $t$ edges
of $G'$, and where $X = N_{H}(Y)$. 
Letting $x^{*}$ and $j^{*}$ 
denote the size of an optimal solution for the rank  reduction and minimum $t$-edge subgraph 
problems, respectively, it follows that $x^{*} = nj^{*} + t$. 

Now suppose that the rank reduction problem admits a $cm^{\eps}$-approximation algorithm, 
where $0 < \eps < 1/2$ and $c > 0$ are absolute constants. 
Letting $(X, Y)$ be a canonical pair with $|Y|=t$ obtained using this algorithm, we have 
$$
n|G_{Y}| + t = |X| \leq cm^{\eps}x^{*} = c(n^{2}+|E|)^{\eps}(nj^{*} + t) 
\leq c(2n^{2})^{\eps}(nj^{*} + t)
$$ 
and hence
$$
|G_{Y}| \leq \frac{c(2n^{2})^{\eps}(nj^{*} + t) - t}{n}
\leq
2cn^{2\eps}\left(j^{*} + \frac{t}{n}\right)
\leq
2cn^{2\eps}\left(j^{*} + j^{*}\right)
= 
4cn^{2\eps}j^{*}.
$$
(In the last inequality we used the fact that $(j^{*})^{2} \geq t$, and thus
$\frac{t}{n} \leq \frac{(j^{*})^{2}}{n} \leq \frac{j^{*}n}{n}=j^{*}$.)
Therefore, $G_{Y}$ is a $t$-edge subgraph whose order is within a $4cn^{2\eps}$-factor of
optimal. 
\end{proof}

As pointed out to us by an anonymous referee, the following lemma is implicit 
in the recent work of Chlamtac, Dinitz, and Krauthgamer~\cite{CDK12} 
(in~\cite{CDK12}, 
the minimum $t$-edge subgraph problem is called the {\DEF smallest $m$-edge subgraph} problem). 
We include a proof nevertheless, for completeness. 

\begin{lemma}
\label{lem:to_densest_k}
For each constant $\eps$ with $0 < \eps < 1/2$, 
every $O(n^{\eps})$-approximation algorithm for the minimum $t$-edge subgraph problem 
can be turned into an $O(n^{2\eps})$-approximation algorithm for the densest $k$-subgraph problem. 
\end{lemma}
\begin{proof}
Suppose that the minimum $t$-edge subgraph problem admits a $c'n^{\eps}$-approximation algorithm, which we denote 
$\mathcal{A}$, where $0 < \eps < 1$ and $c' \geq 1$ are absolute constants. 
Let $G$ be an instance of the densest $k$-subgraph problem. As before, we let $n$ and $m$ 
denote the number of vertices and edges of $G$, respectively.  
We may assume $n \geq k \geq 2$.  
Since $c' \geq 1$ and $n^{\eps} \geq 1$, there exists $c$ 
with $c' \leq c \leq 2c'$ such that $cn^{\eps}$ is an integer. We
will consider the approximation factor of $\mathcal{A}$ to be $cn^{\eps}$ in what follows, 
to avoid cumbersome floors and ceilings in the calculations.  

Run algorithm $\mathcal{A}$ on $G$  with $t = 1, 2, \dots, m$. Let $H_{i}$ be the $i$th 
subgraph returned by the algorithm. Clearly, we may suppose that $|H_{i}| \leq |H_{i+1}|$ 
for each $i\in \{1, \dots, m-1\}$. 

Let $z^{*}$ denote the number of edges in an optimal solution to the densest $k$-subgraph problem on $G$.  
If $|H_{m}| \leq kcn^{\eps}$ then let $t' := m$, otherwise let $t'$ be the index in 
$\{1, \dots, m-1\}$ such that 
$|H_{t'}| \leq kcn^{\eps}$ and $|H_{t'+1}| > kcn^{\eps}$. 
Since algorithm $\mathcal{A}$ is a $cn^{\eps}$-approximation algorithm, and since either $t'=m$ or
$|H_{t'+1}| > kcn^{\eps}$, it follows that every subgraph of $G$ with exactly $k$ vertices has 
at most $t'$ edges, that is, $z^{*} \leq t'$. 

Let $q:= \left\lceil |H_{t'}| / \lfloor k/2 \rfloor \right\rceil$.  
Observe that $2 \leq q \leq  3cn^{\eps}$. 
Let $V_{1}, V_{2}, \dots, V_{q}$ be a partition of the vertex
set of $H_{t}$ into $q$ subsets with $|V_{1}|=|V_{2}|=\cdots =|V_{q-1}|=\lfloor k/2 \rfloor$ 
and $|V_{q}| = |H_{t'}| - (q-1)\lfloor k/2 \rfloor \leq \lfloor k/2 \rfloor$. 
Let $(V_{i}, V_{j})$ be a pair with $i\neq j$ such that $||H_{t'}[V_{i} \cup V_{j}]||$ is maximized. 
By the pigeonhole principle, 
$$
||H_{t'}[V_{i} \cup V_{j}]|| \geq \frac{||H_{t'}||}{{q \choose 2}} = 
\frac{t'}{{q \choose 2}} \geq \frac{t'}{(3cn^{\eps})^{2}} \geq \frac{z^{*}}{9c^{2}n^{2\eps}}. 
$$  
If $|V_{i} \cup V_{j}| = k$, then let $H:=H_{t'}[V_{i} \cup V_{j}]$. If, on the other hand, 
$|V_{i} \cup V_{j}| < k$, then let $H:=H_{t'}[V_{i} \cup V_{j} \cup X]$ 
where $X$ is an arbitrary subset of $V(H_{t'}) \setminus (V_{i} \cup V_{j})$ of 
size $k - |V_{i} \cup V_{j}|$. 
Thus in both cases $|H| = k$ and $||H|| \geq z^{*} / 9c^{2}n^{2\eps}$. Hence,  
$H$ is a solution to the densest $k$-subgraph problem on $G$ whose number of edges is  
within a $9c^{2}n^{2\eps}$-factor of the optimum. 
\end{proof}

Combining Lemma~\ref{lem:to_t_edge} and~\ref{lem:to_densest_k} gives:

\begin{theorem}
\label{th:densest_k_subgraph}
For each constant $\eps$ with $0 < \eps < 1/4$, 
every $O(m^{\eps})$-approximation algorithm for the rank reduction problem on transversal matroids  
with $m$ elements can be turned into an $O(n^{4\eps})$-approximation algorithm for the densest $k$-subgraph problem on graphs with $n$ vertices.
\end{theorem}

As discussed in the introduction, the best approximation algorithm for the 
densest $k$-subgraph problem currently known has an approximation ratio of $O(n^{1/4 + \delta})$ 
for any fixed $\delta > 0$~\cite{BCC10} and it is conjectured that the inapproximability of the problem is of a 
similar magnitude. It would be nice to obtain strong inapproximability bounds for the rank reduction problem that
do not rely on this conjecture. One approach may be to analyze hypergraphs
as the rank reduction problem in transversal matroids incorporates the hypergraph version of 
the minimum $t$-edge subgraph problem. That is, we wish to select as few vertices as possible that induce 
at least $t$ hyperedges. Perhaps surprisingly, little is known about this problem.
As far as we are aware, the only specific hardness result is NP-hardess due Vinterbo \cite{Vin02} 
who studied the problem in the context of making medical databases anonymous.

We conclude this section with a remark about the approximability of the minimum $t$-edge subgraph problem itself. 
Given the existence of a $O(n^{1/4 + \delta})$-approximation algorithm 
for the densest $k$-subgraph problem, in view of Lemma~\ref{lem:to_densest_k}  it is perhaps 
natural to wonder whether one could achieve a $O(n^{1/8 + \delta})$-approximation for the former problem. 
While this is still open as far as we know,  
Chlamtac {\it et al.}~\cite{CDK12} recently made progress in that direction by describing an 
algorithm for the minimum $t$-edge subgraph problem with an approximation ratio of $O(n^{3 - 2\sqrt{2} + \delta})= O(n^{0.1716 + \delta})$ for fixed $\delta > 0$.

\section{The Maximum Vertex Cover Problem in Bipartite Graphs}
\label{sec:mvc}

As we have seen, the rank reduction problem admits a 
fairly simple polynomial-time algorithm on partition matroids but becomes much 
harder on transversal matroids, in the sense that approximation algorithms offering  
good guarantees seem unlikely to exist.  
Another interesting generalization of the case of partition matroids is to 
consider the intersection of two partition matroids. 

As is well-known, the set of matchings of a bipartite graph $G=(V,E)$ with bipartition $(A, B)$ 
can be modeled as the family of common independent sets of two partition matroids 
$M_{1}$ and $M_{2}$ defined on $E$: Take $M_{1}$ to be the partition matroid with model 
$((E(u_{1}), 1), \dots, (E(u_{a}), 1))$ and $M_{2}$
the partition matroid with model $((E(v_{1}), 1), \dots, (E(v_{b}), 1))$, where 
$A=\{u_{1}, \dots, u_{a}\}$, $B=\{v_{1}, \dots, v_{b}\}$, and for $w\in V$ the set $E(w)$ 
denotes the set of edges incident to $w$. 
Hence, in this specific case the rank reduction problem on $M_{1} \cap M_{2}$ amounts to 
finding a subset $F$ of edges of $G$ of minimum size such that  
$\mu(G-F) \leq \mu(G) - k$. 
In this section we show that this problem is NP-hard. More accurately,  
we show that a problem polynomially equivalent to it, the {\mvc} problem 
on bipartite graphs, is NP-hard; see Theorem~\ref{th:mvc_bip}. This solves an 
open problem of B.\ Simeone (see~\cite{EgresOpen}). 

The {\DEF {\mvc}} problem (also known as the {\DEF partial vertex cover} problem) 
is defined as follows: Given a graph $G=(V, E)$
and a positive integer $k \leq |G|$, find a subset $X$ of vertices
of $G$ with $|X|=k$ such that the number of edges covered by $X$ is maximized.
(An edge $e$ of $G$ is {\DEF covered} by $X$ if $e$ has at least one endpoint in $X$.) 

Now, if $G$ is bipartite, $t$ is a positive integer with $t \leq \mu(G)$, and $F$ is a
subset of edges of $G$ such that $\mu(G - F) \leq \mu(G) - t =: k$, 
then by K\H{o}nig's theorem $G - F$ has a vertex cover $X$ of size $k$, and hence $X$ covers at least
$|E| - |F|$ edges of $G$. (We remark that $X$ could cover some edges of $F$ too, and that 
$X$ can be computed in polynomial time given $F$.).    
Conversely, for every set $X \subseteq V$ with $|X| = k$, the set $F$ of edges of $G$ {\em not} covered 
by $X$ is such that $\mu(G - F) \leq k$. 
Therefore, for bipartite graphs, the {\mvc} problem is polynomially 
equivalent to that of finding a minimum-size set of edges decreasing the matching number 
by a prescribed amount. 

It should be noted that two recent works~\cite{DP2011, BGLMPP2012} with an overlapping set of authors  
claim that the NP-hardness of the {\mvc} problem on bipartite graphs 
can be derived directly from the reduction of Corneil and Perl~\cite{CP1984} 
showing that the densest $k$-subgraph problem is NP-hard on bipartite graphs. 
However, the argument relating the latter reduction to the {\mvc} problem, 
described explicitly in~\cite[Lemma~4]{DP2011}, is flawed.\footnote{As mentioned 
in~\cite[Lemma~3]{DP2011}, the {\mvc} problem in bipartite graphs is polynomially 
equivalent to the densest $k$-subgraph problem in complements of bipartite graphs. 
Thus one may equivalently consider the complexity of the latter problem. 
In the proof of Lemma~4 in~\cite{DP2011}, the authors point 
out that the reduction of~\cite{CP1984} implies
that the problem of finding a densest $k$-subgraph in the complement of a bipartite
graph with bipartition $(A, B)$ with $k = q + {q \choose 2}$ and with the extra requirement that 
it contains exactly $q$ vertices from $A$ and ${q \choose 2}$ vertices from $B$ is NP-hard. 
From this they wrongly conclude that the densest $k$-subgraph problem, without this extra constraint, 
is also NP-hard on complements of a bipartite graphs.  
(In fact, the instances obtained via the reduction in~\cite{CP1984} satisfy $|B| \geq k$, and thus 
a densest $k$-subgraph is trivially obtained by taking $k$ vertices in the clique $B$.) 
} 
We also mention that the proof of Theorem~1 in~\cite{BGLMPP2012}, showing that 
a related problem called the maximum quasi-independent set problem
is NP-hard on bipartite graphs, relies on the assumption that the
{\mvc} problem is NP-hard on bipartite graphs. Thus 
our result also fills a gap in that proof. 

\begin{theorem}
\label{th:mvc_bip}
The {\mvc} problem is NP-hard on bipartite graphs.
\end{theorem}

Before proving Theorem~\ref{th:mvc_bip}, we need to introduce a technical lemma.

\begin{lemma}
\label{lem:IP}
Let $\l$ be an integer with $\l \geq 6$. Then the integer program
$$
\begin{array}{ll}
\textrm{\em \bf minimize} \quad &  x + 2y + 3z   \\[1ex]
\textrm{\em s.\ t.}  & \displaystyle  x + y + z - s = {\l \choose 2} - \l   \\[0.7ex]
& x \leq \displaystyle {s \choose 2}   \\[1.1ex]
& \displaystyle x, y, z, s \in \N
\end{array}
$$
has a unique optimal solution given by $x = {\l \choose 2}, y = 0, z = 0, s = \l$.
\end{lemma}
\begin{proof}
The proof is a straightforward case analysis.
Consider an optimal solution $x,y,z,s$ to the integer program and, arguing by contradiction, assume
it differs from the solution described above.
Let $f(x,y,z) := x + 2y + 3z$.

{\bf Case 1: $s \leq 2$}. We have $x \leq s$  
and thus
$$
y + z = {\l \choose 2} -\l + s - x 
\geq {\l \choose 2} -\l. 
$$
It follows that $f(x,y,z) \geq 2(y+z) \geq 2{\l \choose 2} - 2\l$.
But $2{\l \choose 2} - 2\l > {\l \choose 2}$ since $\l \geq 6$, contradicting the optimality
of the solution.

{\bf Case 2: $3 \leq s < \l$}. 

$$
y + z = {\l \choose 2} -\l + s - x 
\geq  {\l \choose 2} -\l +s -{s \choose 2} 
= \left({\l \choose 2} -  {s \choose 2}\right) - (\l-s) 
\geq 2.
$$
Here, the last inequality follows from the fact that ${s+1 \choose 2} - {s \choose 2} \ge 3$.

Now, increment $s$ by $1$, $x$ by $3$, and decrement $y$ and $z$
in such a way that they remain non-negative integers and that the sum $y+z$ decreases by exactly $2$.
The modified solution is still feasible and $f(x,y,z)$ decreases by at least $1$, a contradiction.

{\bf Case 3: $s = \l$}. Then $y + z \geq 1$, since otherwise we would have the solution described in the lemma statement. 
It follows that $x \leq {\l \choose 2} - 1$. Thus we can increment $x$ by $1$ and
decrement by $1$ a positive variable among $y,z$. This strictly decreases $f(x,y,z)$, a contradiction.

{\bf Case 4: $s > \l$}. Then $x+y+z={\l \choose 2} +s -l > {\l \choose 2}$. But $x\le {\l \choose 2} < {s \choose 2}$, otherwise 
the solution cannot be minimum. Therefore $y+z\ge 1$. Thus we improve the solution by incrementing $x$ by $1$ and
decrementing by $1$ a positive variable among $y,z$. 
\end{proof}

Now we may turn to the proof of Theorem~\ref{th:mvc_bip}.

\begin{proof}[Proof of Theorem~\ref{th:mvc_bip}.]
The reduction is from the NP-complete problem 
{\DEF Clique}: Given a graph $H$ and an integer $\l$, decide whether $H$ contains a
clique on $\l$ vertices or not. We may assume $\l \geq 6$ (otherwise, we simply 
check the existence of an $\l$-clique by brute force). We may also suppose that $H$
has minimum degree at least $2$. Indeed, a vertex with degree at most $1$ cannot be part of 
an $\l$-clique, and thus those vertices can iteratively be removed from the graph. 
Finally, we assume that 
$||H|| \geq |H| + {\l \choose 2}$. This last assumption can also
be made without loss of generality. Indeed, if $||H||$ is too small then
one can simply consider the disjoint union of $H$ with a large enough $3$-regular graph; 
since $\l \geq 6$ no vertex from this new $3$-regular component can be part of an $\l$-clique.

We build an instance $(G, k)$ of the {\mvc} problem as follows. 
First, create two adjacent vertices $a_{v}$ and $b_{v}$ for every vertex $v\in V(H)$, and
similarly two adjacent vertices $a_{e}$ and $b_{e}$ for every edge $e \in E(H)$.
Next, for every edge $uv \in E(H)$, add the edges $a_{uv}b_{u}$, $b_{uv}a_{u}$
and $a_{uv}b_{v}$, $b_{uv}a_{v}$. Finally, let 
$$
k:= |H| + ||H|| - {\l \choose 2} + \l.
$$
Observe that $G$ is bipartite with bipartition
$$
\big(\{a_{x}: x\in V(H) \cup E(H)\}, \{b_{x}: x\in V(H) \cup E(H)\}\big).
$$
See Figure~\ref{fig} for an illustration of the construction. 

\begin{figure}
\centering
\includegraphics[width=0.7\textwidth]{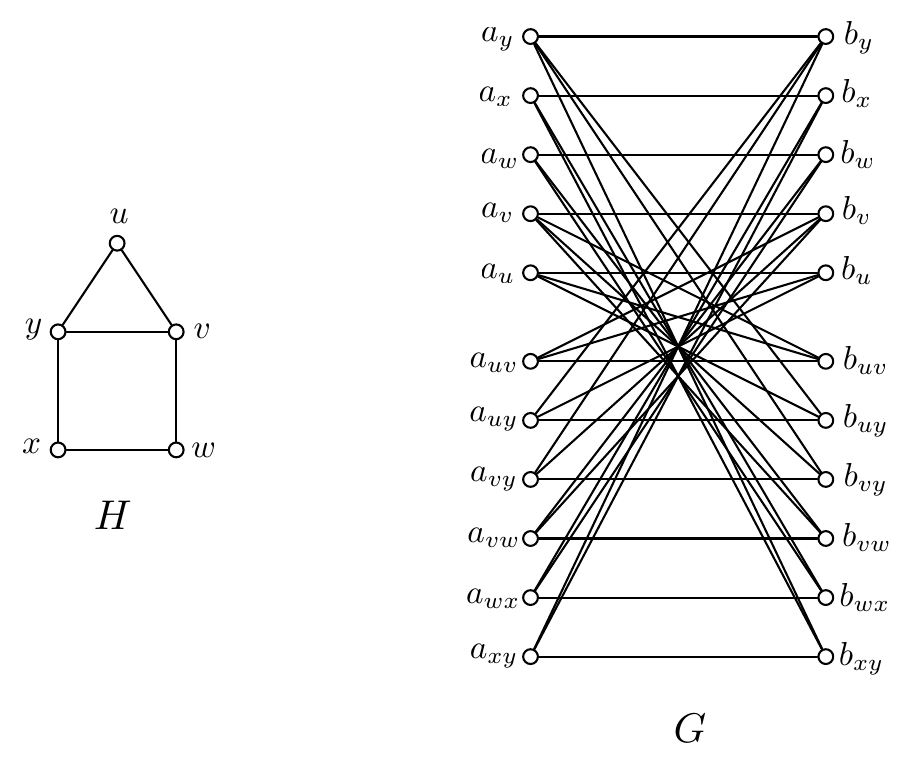}
\caption{\label{fig}Illustration of the construction of the bipartite graph $G$. 
(We note that this graph $H$ cannot be a valid instance of the problem because $\ell \geq 6$ 
and thus $H$ cannot satisfy $||H|| \geq |H| + {\l \choose 2}$; our aim here 
is only to illustrate how $G$ is obtained from $H$ on a small example.)}
\end{figure}

A feasible solution for this instance of the {\mvc} problem is a 
subset $X$ of vertices of $G$ with $|X|=k$, which we call a {\DEF {\pvc}} for short.
We let $c(X)$ denote the number of edges covered by such a set $X$. 
Let $OPT$ denote the maximum of $c(X)$ over every {\pvc} $X$ of $G$.

A {\pvc} $X$ of $G$ is {\DEF nice} if 
$$
X \cap \{a_{u}, b_{u}\} \in \big\{\{a_{u}\}, \{a_{u}, b_{u}\}\big\}
$$
for every $u\in V(H)$ and
$$
X \cap \{a_{e}, b_{e}\} \in \big\{\varnothing, \{a_{e}\}\big\}
$$
for every $e\in E(H)$.

\begin{claim}
\label{claim:nice}
Given a {\pvc} $X$ of $G$, one can find a nice {\pvc} $X'$ of $G$
with $c(X') \geq c(X)$.
\end{claim}
\begin{proof}
First we define a {\pvc} $\tilde X$ based on $X$ which is close to being nice:
Let 
\begin{align*}
\tilde X := \, &
\big\{a_{u}, b_{u}: u\in V(H), a_{u}, b_{u} \in X \big\} \cup
\big\{a_{u}: u\in V(H), |\{a_{u}, b_{u}\} \cap X| = 1 \big\}  \\
& \cup
\big\{a_{e}, b_{e}: e\in E(H), a_{e}, b_{e} \in X \big\} \cup
\big\{a_{e}: e\in E(H), |\{a_{e}, b_{e}\} \cap X| = 1 \big\}.
\end{align*}

By construction $|\tilde X| = |X|$. 
Clearly, an edge $a_{x}b_{x}$ with $x\in V(H) \cup E(H)$ is covered by $\tilde X$
if and only if it is covered by $X$. Also, given a pair $(u, e)$ of vertex
$u \in V(H)$ and edge $e\in E(H)$ such that $u$ is incident to $e$ in $H$,
the set $\tilde X$ covers at least as many edges 
in $\{a_{e}b_{u}, b_{e}a_{u}\}$ as $X$ (though not necessarily the same ones).
It follows that $c(\tilde X) \geq c(X)$.

A useful property of the set $\tilde X$ is that if $b_{x}\in \tilde X$ for some 
$x \in V(H) \cup E(H)$ then necessarily  $a_{x}\in \tilde X$. For simplicity
we call this property the {\DEF $a$-property} of $\tilde X$.

We need to introduce an additional definition.
An element $x\in V(H) \cup E(H)$ is said to be {\DEF bad} in a 
{\pvc} $Y$ of $G$
if either $x\in V(H)$ and $a_{x}, b_{x}\notin Y$ ($x$ is a bad vertex), 
or $x\in E(H)$ and $a_{x}, b_{x}\in Y$ ($x$ is a bad edge). 
Observe that $Y$ is nice if and only if $Y$ has the $a$-property
and there is no bad element. 

Suppose $e=uv$ is an edge of $H$ which is bad in $\tilde X$.
If $u$ or $v$ is also bad in $\tilde X$, say $u$,
then let 
$$
\tilde X' := \left(\tilde X - \{b_{e}\}\right) \cup \{a_{u}\}.
$$
We have $a_{e} \in \tilde X$, thus the edge $a_{e}b_{e}$ is still covered by $\tilde X'$.
Since $\tilde X'$ covers also $b_{e}a_{u}$, there is at most one edge incident to
$b_{e}$ in $G$ (namely, $b_ea_v$) which is not covered by $\tilde X'$. On the other hand, $\tilde X'$ covers
the previously uncovered edge $a_{u}b_{u}$. Hence, 
$\tilde X'$ is a {\pvc} with  $c(\tilde X') \geq c(\tilde X)$. Observe that $\tilde X'$
still has the $a$-property, and the edge $e$ is no longer bad in $\tilde X'$.

If, on the other hand, none of $u,v$ is bad in $\tilde X$, then $a_{u}, a_{v} \in \tilde X$ by the $a$-property.
Since $a_{e} \in \tilde X$, it follows that $c(\tilde X - \{b_{e}\}) = c(\tilde X)$.
There exists an element $x\in V(H) \cup E(H)$ such that $a_{x} \notin \tilde X$, because
$|\tilde X| = k = |H| + ||H|| - {\l \choose 2} + \l < |H| + ||H||$ (since $\l \geq 6$).
Let then 
$$
\tilde X' := \left(\tilde X - \{b_{e}\}\right) \cup \{a_{x}\}. 
$$
The set $\tilde X'$
is a {\pvc} with the $a$-property and with $c(\tilde X') \geq c(\tilde X)$. 
Moreover, the edge $e$ is no longer bad in $\tilde X'$.

Now apply iteratively the above modifications on $\tilde X$ as
long as there exists a bad edge. This results in
a {\pvc} $\widehat X$ with the $a$-property, without bad edges,
and with $c(\widehat X) \geq c(\tilde X)$.

Next we deal with bad vertices in $\widehat X$.
Suppose $u\in V(H)$ is such a vertex, that is, $a_{u}, b_{u} \notin \widehat X$.
Consider two edges $e,f$ incident to $u$ in $H$. 
(Recall that $H$ has minimum degree at least $2$.)
Since $|\widehat X| = k = |H| + ||H|| - {\l \choose 2} + \l$ and 
$||H|| \geq |H| + {\l \choose 2}$ by our assumption on $H$, we have
$| \widehat X| > 2|H|$. Together with the $a$-property of $\widehat X$,
it follows that $a_{e'} \in \widehat X$ for some edge $e'\in E(H)$
(possibly $e'=e$ or $e'=f$). Note that $b_{e'} \notin \widehat X$, because otherwise
$e'$ would be a bad edge for $\widehat X$. Let
$$
\widehat X' := \left(\widehat X - \{a_{e'}\}\right) \cup \{a_{u}\}.
$$
Since $a_{e'}$ has degree $3$ in $G$ we have $c(\widehat X - \{a_{e'}\}) \geq c(\widehat X) - 3$.
Furthermore, $b_e, b_f \notin \widehat X$ as there are no bad edges in  $\widehat X$.
Thus, the three edges $a_{u}b_{u}, b_{e}a_{u}, b_{f}a_{u}$ of $G$ were not covered
by $\widehat X$ but are covered by $\widehat X'$, so we have 
$c(\widehat X') \geq c(\widehat X)$. Similarly as before, the {\pvc} $\widehat X'$
has the $a$-property and one less bad vertex than $\widehat X$.
Therefore, iterating this procedure as long as there is a bad vertex,
we eventually obtain a {\pvc} $X'$ with $c(X') \geq c(\widehat X) \geq c(X)$
having the $a$-property and no bad element, as desired.
\end{proof}

Consider a nice {\pvc} $X$ of $G$. 
Let $S(X)$ be the set of vertices $u\in V(H)$ such that 
$a_{u}, b_{u} \in X$, and let $s(X) := |S(X)|$.
An edge $e=uv$ of $H$ 
satisfies exactly one of the following three conditions: 
\begin{enumerate}
\item $u,v \in S(X)$;
\item exactly one of $u, v$ belongs to $S(X)$,
\item $u,v \notin S(X)$.
\end{enumerate}
We say that edge $e$ is of {\DEF type $i$} ($i \in \{1,2,3\}$)
if $e$ satisfies the $i$th condition above {\em and} moreover $a_{e}, b_{e} \notin X$. 
(We will focus on edges $e$ of $H$ such that $a_{e}, b_{e} \notin X$ in what follows, 
which is why the other ones do not get assigned a type.) 
Let $E_{i}(X)$ be the set of edges of $H$ with type $i$, and let $e_{i}(X) := |E_{i}(X)|$.

\begin{claim}
\label{claim:E}
Let $X$ be a nice {\pvc}. 
Then
$$
c(X) = ||G|| - e_{1}(X) - 2e_{2}(X) - 3e_{3}(X).
$$
\end{claim}
\begin{proof}
As $X$ is nice, $a_v\in X$ for all $v\in V$. 
Therefore every edge of the form $a_vb_v$ or $a_vb_e$ is covered.
Also, for each edge $e=uv \in E(H) \setminus \cup_{1\leq i \leq3} E_{i}(X)$ we have that
$a_eb_e, a_eb_u, a_eb_v$ are all covered. Thus the only uncovered 
edges are incident to vertices $a_e$ where $e\in \cup_{1\leq i \leq3} E_{i}(X)$. 
Suppose $e\in \cup_{1\leq i \leq3} E_{i}(X)$ and let us consider 
which edges among the three edges $a_eb_e, a_eb_u, a_eb_v$ are covered by $X$. 
If $e \in E_{1}(X)$, then $X$ covers $a_eb_u$ and $a_eb_v$ but not $a_eb_e$. 
If $e \in E_{2}(X)$, then $X$ covers exactly one of $a_eb_u$, $a_eb_v$, and avoids $a_eb_e$. 
If $e \in E_{3}(X)$, then $X$ covers none of the three edges. 
Hence, the total number of edges not covered by $X$ is exactly 
$e_{1}(X) + 2e_{2}(X) + 3e_{3}(X)$. 
\end{proof}

\begin{claim}
\label{claim:cost}
Let $X$ be a nice {\pvc}. Then
$$
c(X) \leq ||G|| - {\l \choose 2},
$$
with equality if and only if $s(X) = \l$, $e_{1}(X)= {\l \choose 2}$, $e_{2}(X)=e_{3}(X)=0$.
\end{claim}
\begin{proof}
Let $x:= e_{1}(X)$, $y := e_{2}(X)$, $z := e_{3}(X)$ and $s:=s(X)$.
Then $c(X) = ||G|| - f(x,y,z)$, where $f(x,y,z) = x + 2y +3z$ by the previous claim. 

Every edge in $E_{1}(X)$ has its two endpoints in $S(X)$; hence,
\begin{equation}
\label{eq:cond1}
{s \choose 2} \geq x. 
\end{equation}
Also, 
$$
|H| + ||H|| - (x+y+z-s) =  |X| = k = |H| + ||H|| - \left({\l \choose 2} - \l\right),
$$
and thus 
\begin{equation}
\label{eq:cond2}
x+y+z-s = {\l \choose 2} - \l.
\end{equation}

Since $\l \geq 6$ and $x, y,z, s,\l$ are non-negative integers satisfying \eqref{eq:cond1}
and \eqref{eq:cond2}, by Lemma~\ref{lem:IP} we have
$f(x,y,z) \geq {\l \choose 2}$, or equivalently, 
$$
c(X) \leq ||G|| - {\l \choose 2}.
$$
Moreover, equality holds if and only if $x = {\l \choose 2}, y = 0, z = 0, s = \l$
by the same lemma.
\end{proof}

It follows from Claims~\ref{claim:nice} and \ref{claim:cost} that 
$$
OPT \leq ||G|| - {\l \choose 2}.
$$

If $H$ has an $\l$-clique $K \subseteq H$, then the subset $X \subseteq V(G)$
defined by
\begin{align*}
X:=\, & \{a_{u}, b_{u}: u\in V(K)\} 
\cup \{a_{u}: u\in V(H) - V(K)\}   \\
&\cup \{a_{e}: e\in E(H) - E(K)\}  
\end{align*}
is a {\pvc} of $G$ with $c(X) = ||G|| - {\l \choose 2}$, implying $OPT = ||G|| - {\l \choose 2}$.

Conversely, if $OPT = ||G|| - {\l \choose 2}$, then there exists a {\pvc} $X$ of $G$
with $c(X) = ||G|| - {\l \choose 2}$, and by Claim~\ref{claim:nice} we may
assume that $X$ is nice. From Claim~\ref{claim:cost} we then have
$s(X) = \l$, $e_{1}(X)= {\l \choose 2}$, $e_{2}(X)=e_{3}(X)=0$,
implying that $S(X)$ induces an $\l$-clique in $H$. 

Therefore, we can decide in polynomial time if $H$ has an $\l$-clique 
by checking if $OPT = ||G|| - {\l \choose 2}$. This concludes the proof.
\end{proof}

The NP-hardness of the {\mvc} problem on bipartite graphs motivates the search for 
non-trivial approximation algorithms for this class of graphs. A recent result in this direction is due to 
Apollonio and Simeone~\cite{APapprox}, who gave an LP-based $(4/5)$-approximation algorithm for bipartite graphs.  

\section*{Acknowledgments}
We are grateful to Attila Bernath and Tam\'as Kiraly for 
their comments on the proof of Theorem~\ref{th:mvc_bip}, and to the two anonymous referees for their helpful remarks and suggestions. 
We also thank Mohit Singh, Bill Cunningham and Anupam Gupta for interesting discussions, 
and Nicola Apollonio for providing us with a preliminary version of~\cite{APsub}.

\end{document}